\documentclass[11pt,letterpaper]{article}

\usepackage{fullpage}
\usepackage{amsthm}
\usepackage{amssymb,amsmath,amsfonts}
\usepackage{enumerate}
\usepackage[style=alphabetic,sorting=nyt,maxnames=20,maxcitenames=4,maxbibnames=20,backend=bibtex]{biblatex}
\usepackage{authblk}
\usepackage{hyperref}
\usepackage[usenames,dvipsnames]{color}

\theoremstyle{plain}
\newtheorem{theorem}{Theorem}[section]

\newtheorem{lemma}[theorem]{Lemma}
\newtheorem{corollary}[theorem]{Corollary}
\newtheorem{definition}[theorem]{Definition}
\newtheorem{conjecture}[theorem]{Conjecture}

\theoremstyle{definition}
\newtheorem*{remark}{Remark}

\newtheorem{innercustomthm}{Theorem}
\newenvironment{customthm}[1]
  {\renewcommand\theinnercustomthm{#1}\innercustomthm}
  {\endinnercustomthm}

\newenvironment{informalthm}
  {\medskip\noindent{\bf Theorem.}}
  {\medskip}

\newcommand{\NN}{\mathbb{N}}

\newcommand{\RR}{\mathbb{R}}

\newcommand{\EE}{\mathop{{}\mathbb{E}}}
\newcommand{\rarr}{\rightarrow}
\newcommand{\larr}{\leftarrow}

\newcommand{\eps}{\varepsilon}

\newcommand{\defeq}{\overset{{\rm def}}{=}}

\newcommand{\Adv}{\mathcal{A}}

\newcommand{\negl}{\mathsf{negl}}
\newcommand{\poly}{\mathsf{poly}}
\newcommand{\polylog}{\mathsf{polylog}}

\newcommand{\statIndist}{\overset{s}{\approx}}

\newcommand{\rprimesblue}{\textcolor{NavyBlue}{r'_1,\dots,r'_s}}
\newcommand{\rprimesred}{\textcolor{WildStrawberry}{r''_1,\dots,r''_s}}
\newcommand{\rhats}[2]{\hat{r}_{#1},\dots,\hat{r}_{#2}}
\newcommand{\rs}[2]{r_{#1},\dots,r_{#2}}

\newcommand{\cM}{{\cal M}}

\definecolor{commentColor}{rgb}{0,0.6,0.8}

\begin{document}

\title{Adaptively Secure Coin-Flipping, Revisited}

\author[1]{Shafi Goldwasser}
\affil[1]{MIT and the Weizmann Institute of Science}
\author[2]{Yael Tauman Kalai}
\affil[2]{Microsoft Research}
\author[3]{Sunoo Park}
\affil[3]{MIT}

\date{}
\maketitle

\begin{abstract}
The full-information model was introduced by Ben-Or and Linial in 1985 to study collective coin-flipping: the problem of generating a common bounded-bias bit in a network of $n$ players with $t=t(n)$ faults. They showed that the majority protocol, in which each player sends a random bit and the output is the majority of the players' bits, can tolerate $t(n)=O (\sqrt n)$ even in the presence of \emph{adaptive} corruptions, and they conjectured  that this is optimal for such adversaries. Lichtenstein, Linial, and Saks proved that the conjecture holds for protocols in which each player sends only a single bit. Their result has been the main progress on the conjecture during the last 30 years.

In this work we revisit this question and ask: what about protocols where players can send longer messages? Can increased communication allow for a larger fraction of corrupt players?

We introduce a model of \emph{strong adaptive} corruptions, in which an adversary sees all messages sent by honest parties in any given round and, based on the message content, decides whether to corrupt a party (and alter its message or sabotage its delivery) or not. This is in contrast to the (classical) adaptive adversary who can corrupt parties only based on past messages, and cannot alter messages already sent. 

We prove that any one-round coin-flipping protocol, \emph{regardless of message length}, can be secure against at most $\widetilde{O}(\sqrt n)$ strong adaptive corruptions. Thus, increased message length does not help in this setting.

We then shed light on the connection between adaptive and strongly adaptive adversaries, by proving that for any symmetric one-round coin-flipping protocol secure against $t$ adaptive corruptions, there is a symmetric one-round coin-flipping protocol secure against $t$ strongly adaptive corruptions. Going back to the standard adaptive model, we can now prove that any symmetric one-round protocol with arbitrarily long messages can tolerate at most $\widetilde{O}(\sqrt n)$ adaptive corruptions. 

At the heart of our results there is a novel use of the Minimax Theorem and a new technique for converting any one-round secure protocol  with arbitrarily long messages into a secure one where each player sends only  $\polylog(n)$ bits. This technique may be of independent interest.
\end{abstract}


\section{Introduction}

A collective coin-flipping protocol is one where a set of~$n$ players use private randomness to generate a common random bit~$b$.  Several protocol models have been studied in the literature.  In this work, we focus on the  model of {\em full information} \cite{BL85} where all parties communicate via a single broadcast channel.  

The challenge is that $t=t(n)$ of the parties may be corrupted and aim to bias the protocol outcome (i.e. the ``coin'') in a particular direction.  We focus on {\em Byzantine faults}, where once a party is corrupted, the adversary completely controls the party and can send any message on its behalf. Two types of Byzantine adversaries have been considered in the literature: {\em static} adversaries and {\em adaptive} adversaries.  A static adversary is one that chooses which~$t$ players to corrupt {\em before the protocol begins}.  An adaptive adversary is one who may choose which~$t$ players to corrupt adaptively, as the protocol progresses. 

Collective coin-flipping in the case of static adversaries is well understood (see section~\ref{related-work}).   
In this work, our focus is on the setting of adaptive adversaries, which has received considerably less attention. 
A collective coin-flipping protocol is said to be secure against~$t$ adaptive (resp. static) corruptions if for any adaptive adversary corrupting~$t$ parties, there is a constant $\eps>0$ such that the probability that the protocol outputs~$0$ (and the probability that the protocol outputs~$1$) is at least $\eps$, where the probability is taken over the randomness of the players and the adversary.  

The question we study is: {\em What is the maximum number of adaptive corruptions that a secure coin-flipping protocol can tolerate?} On the positive side, it has been shown by Ben-Or and Linial \cite{BL85} in 1985 that 
the majority protocol (where each party sends a random bit, and the output is equal to the majority of the bits sent), is resilient to~$\Theta(\sqrt n)$ adaptive corruptions.
Ben-Or and Linial conjectured that this is in fact optimal.
\begin{conjecture}[\cite{BL85}]\label{conj:coinFlipping}
Majority is the optimal coin-flipping protocol against adaptive adversaries.
In particular, any coin-flipping protocol is resilient to at most $O(\sqrt n)$ adaptive corruptions.
\end{conjecture}

Shortly thereafter,
Lichtenstein, Linial, and Saks \cite{LLS89} proved the conjecture for a restricted class of protocols: namely, those in which each player sends only a single bit.
Their result has been the main progress on the conjecture of \cite{BL85} during the last 30 years.

\subsection{Our contribution}

We first define a new adversarial model of {\em strong adaptive} corruptions.  
Informally, an adversary is strongly adaptive if he can corrupt players depending on the content of their messages. More precisely, in each round,
he can see all the messages that honest players ``would'' send, and then decide which of them to corrupt. 
This is in contrast to a (traditionally defined) adaptive adversary who can, at any point in the protocol, corrupt any player who has not yet spoken based on the history of communication, but cannot alter the message of a player who has already spoken.
Thus, strong adaptive adversaries are more powerful than adaptive adversaries.

We believe that the notion of \emph{strong adaptive security} gives rise to a natural and interesting new adversarial model
in which to study multi-party protocols in general.
Indeed, it is a realistic concern in many settings that malicious parties may decide to stop or alter messages sent by honest players \emph{depending on message content},
and it is a shortcoming that existing adversarial models fail to take such behavior into account.

We consider our strong adaptive adversarial notion to be closely tied to the notion of a \emph{rushing} adversary in the setting of static corruptions.
A rushing static adversary can see the messages that the honest players
send in each round, before deciding the messages that the corrupted players will send in the same round.
The intuitive idea of a rushing adversary is that the adversary \emph{sees all possible information in each round, before making his move}.
We remark that a notion of ``rushing adaptive adversary'' has been previously proposed in the literature, but such an adversary is weaker
than our strong adaptive adversary\footnote{In particular, the ``rushing adaptive adversary'' from the literature can decide the order in which players
send messages in a round, and can decide to corrupt a player who has not yet sent a message within a round. However, unlike our strong adaptive adversary,
this adversary cannot decide to corrupt a player based on the content of the message which the player \emph{would send if uncorrupted}.}.
We argue that our strong adaptive adversary better captures the idea that the adversary \emph{sees all possibly relevant information in each round, before making his move}, since in the adaptive setting, the adversary's strategy must decide not only what messages to send, but also \emph{which players to corrupt}.

Our main result is that the conjecture of \cite{BL85} holds (up to polylogarithmic factors)
for any one-round coin-flipping protocol in the presence of \emph{strong adaptive} corruptions.

\begin{informalthm}
Any secure one-round coin-flipping protocol $\Pi$ can tolerate at most $t=\widetilde{O}(\sqrt n)$ strong adaptive corruptions.
\end{informalthm}

This is shown by a generic reduction of communication in the protocol: first, we prove that any strongly adaptively secure protocol $\Pi$ can be converted to one
where players send messages of no more than polylogarithmic length, while preserving the number of corruptions that can be tolerated.
Then, we show that any protocol with messages of polylogarithmic length can be converted to one where each player sends only a single bit,
at the cost of a polylogarithmic factor in the number of corruptions.
Finally, we reach the \emph{single-bit} setting in which the bound of Lichtenstein et al. \cite{LLS89} can be applied to obtain the theorem.
We believe that our technique of converting any protocol into one with short messages is of independent interest and will find other applications.

Furthermore, we prove that strongly adaptively secure protocols are a more general class of protocols than symmetric adaptively secure protocols.
A symmetric protocol $\Pi$ is a one that is oblivious to the order of its inputs: that is, where
for any permutation $\pi:[n]\rightarrow[n]$ of the players, it holds that the protocol outcome $\Pi(r_1,\ldots,r_n)=\Pi(r_{\pi(1)},\ldots,r_{\pi(n)})$ is the same.

\begin{informalthm}
For any symmetric one-round coin-flipping protocol $\Pi$ secure against $t=t(n)$ adaptive corruptions,
there is a symmetric one-round coin-flipping protocol $\Pi'$ secure against $\Omega(t)$ strong adaptive corruptions.
\end{informalthm}

Curiously, this proof makes a novel use of the Minimax Theorem \cite{vN44,Nash} from game theory, in order to take any symmetric, adaptively secure protocol and convert it to a new protocol which is strongly adaptively secure. This technique views the protocol as a zero-sum game between two players $\Adv_0$ and $\Adv_1$, where $\Adv_0$
wins if the protocol outcome is 0 and $\Adv_1$ wins if the outcome is 1. 
We analyze the ``minimax strategy'' in which the players try to minimize their maximum loss, in order to deduce the strong adaptive security of the new protocol.
Whereas some prior works have made use of game theory in the analysis of (two-party) protocols,
this is the first use of these game-theoretic concepts in the \emph{construction} of distributed multiparty protocols.

Finally, using the above results as stepping stones, we return to the classical conjecture of \cite{BL85},  in the model of adaptive adversaries, and show that the conjecture holds (up to polylogarithmic factors)
for any \emph{symmetric} one-round protocol with arbitrarily long messages.

\begin{informalthm}
Any secure symmetric one-round coin-flipping protocol $\Pi$ can tolerate at most $t=\widetilde{O}(\sqrt n)$ adaptive corruptions.
\end{informalthm}

\subsection{Related work}\label{related-work}

The full-information model (also known as the \emph{perfect information model}) was introduced by Ben-Or and Linial \cite{BL85} 
to study the problem of collective coin-flipping when no secret communication is possible between honest players.

\paragraph{In the static setting.}
Protocols for collective coin-flipping in the presence of static corruptions have been constructed in a series of works 
that variously focus on improving the fault-tolerance, round complexity, and/or bias of the output bit.
Feige \cite{Fei99} gave a protocol that is $(\delta^{1.65}/2)$-secure\footnote{A
coin-flipping protocol is $\eps$-secure against~$t$ static corruptions if for any static adversary that corrupts up to~$t$ parties, 
the probability that the protocol outputs~$0$ is at least $\eps$.}
in the presence of $t=(1+\delta)\cdot n/2$ static corruptions for any constant $0<\delta<1$.
Russell, Saks, and Zuckerman \cite{RSZ02} then showed that any protocol that is secure in the presence of linearly many corruptions must
either have at least $(1/2-o(1))\cdot\log^*(n)$ rounds, or communicate many bits per round.

Interestingly, nearly all proposed \emph{multi-round} protocols for collective coin-flipping
first run a \emph{leader election} protocol in which one of the $n$ players is selected as a ``leader'', who then outputs a bit that is taken as the protocol outcome.
We remark that this approach is inherently unsuitable for adaptive adversaries, which can always corrupt the leader after he is elected, 
and thereby surely control the protocol outcome.

\paragraph{In the adaptive setting.}
The study of coin-flipping protocols has been predominantly in the static setting.
The problem of adaptively secure coin-flipping was introduced by Ben-Or and Linial \cite{BL85}
and further examined by Lichtenstein, Linial, and Saks \cite{LLS89} as described in the previous section.
In addition, Dodis \cite{Dod00} proved that through ``black-box'' reductions from \emph{non-adaptive} coin-flipping, it is not possible tolerate
significantly more corruptions than the majority protocol. The definition of ``black-box'' used in \cite{Dod00}
is rather restricted: it only considers sequential composition of non-adaptive coin-flipping protocols, followed by a (non-interactive)
function computation on the coin-flips thus obtained.

\paragraph{In the pairwise-channels setting.}
An adversarial model bearing some resemblance to our strong adaptive adversary model was introduced and analyzed by Hirt and Zikas \cite{HZ10}
in the \emph{pairwise communication channels} model, rather than the full-information model. In their model, the adversary can corrupt a party $P$ based on
some of the messages that $P$ sends within a round, then the adversary controls the rest of $P$'s messages in that round (and for future rounds).
Unlike in our strong adaptive model, the adversary of \cite{HZ10} cannot ``see inside all players' heads'' 
and overwrite arbitrary honest messages based on their content before they are sent.

Interestingly, a separation has been shown between standard adaptive adversaries and the stronger adversaries of Hirt and Zikas:
\cite{HZ10} shows that broadcast is impossible to achieve for $t>n/2$ corruptions in their stronger adversarial model,
whereas Garay et al. \cite{GKKZ11} showed that broadcast is achievable for any $t<n$ corruptions in the standard adaptive adversarial model.

\paragraph{In the computational setting.}
The problem of generating a shared random bit has also been studied in the setting where players are computationally bounded,
and in different communication network models. Blum \cite{Blu81} introduced the coin-flipping problem in the two-player computational setting; and
Goldreich, Micali, and Wigderson \cite{GMW87} subsequently showed that it is possible to efficiently generate a shared bit with negligible bias, in the presence of static adversaries.

Another line of work shows that the existence of any coin-flipping protocol for computationally bounded players which achieves a sufficiently small bias
implies the existence of one-way functions. The latest result in this line of work, due to Berman, Haitner, and Tentes \cite{BHT14}, proves
that if there exists a two-player coin-flipping protocol that achieves any constant bias, then one-way functions exist.

\section{Preliminaries}

We consider coin-flipping protocols in the \emph{full-information model} (also known as the \emph{perfect information model}), 
where $n$ computationally unbounded players communicate via a single broadcast channel. 
The network is synchronized between rounds, but is asynchronized within each round (that is, there is no guarantee on message ordering within a round,
and an adversary can see the messages of all honest players in a round before deciding his own messages).

In this work, we focus on \emph{one-round} protocols, and we consider protocols that terminate (and produce an output) with probability 1.
In particular, we focus on \emph{coin-flipping} protocols, which are defined as follows.

\begin{definition}[Coin-flipping protocol]\label{def:coinFlipping}
A \emph{coin-flipping protocol} $\Pi=\{\Pi_n\}_{n\in\NN}$ is a family of protocols 
where each $\Pi_n$ is a $n$-player protocol which outputs a bit in $\{0,1\}$.
\end{definition}

\paragraph{Notation.}
We write $\statIndist$ for statistical indistinguishability of distributions.
We denote by $\Pr^\Pi(b)$ the probability that an honest execution of $\Pi$ will lead to the outcome $b\in\{0,1\}$.
We denote by $\Pr^{\Pi,\Adv}(b)$ the probability that an execution of $\Pi$ in the presence of an
adversary $\Adv$ will lead to the outcome $b\in\{0,1\}$. 
The probability is over the random coins of the honest players and the adversary.

For one-round protocols, we write $\Pi_n(r_1,\dots,r_n)$
to denote the outcome of the protocol $\Pi_n$ when each player $i$ sends message $r_i$. 
(The vector $(r_1,\dots,r_n)$ is a protocol \emph{transcript}.)

\subsection{Properties of protocols}

\begin{definition}[Symmetric protocol]
A protocol $\Pi$ is \emph{symmetric} if the outcome of a protocol execution is the same 
no matter how the messages within each round are permuted. In particular, a one-round protocol $\Pi$ is symmetric 
if for all $n\in\NN$ and any permutation $\pi\in[n]\rarr[n]$,
$$\Pi_n(r_{1},\dots,r_{n})=\Pi_n(r_{\pi(1)},\dots,r_{\pi(n)}).$$
\end{definition}

We remark, for completeness, that in the multi-round case,
the outcome of a symmetric protocol should be unchanged even if different permutations are applied in different rounds.

\begin{definition}[Single-bit/multi-bit protocol]
A protocol is \emph{single-bit} if each player sends at most one bit over the course of the protocol execution.
Similarly, a protocol is \emph{$m$-bit} if each player sends at most $m$ bits over the course of the protocol execution.
More generally, a protocol which is not single-bit is called \emph{multi-bit}.
\end{definition}

\begin{definition}[Public-coin protocol]
A protocol is \emph{public-coin} if each honest player broadcasts all of the randomness he generates 
(i.e. his ``local coin-flips''), and does not send any other messages.
\end{definition}


\subsection{Adversarial models in the literature}\label{sec:advModels}

The type of adversary that has been by far the most extensively studied in the coin-flipping literature
is the static adversary, which chooses a subset of players to corrupt \emph{before} the protocol execution begins,
and controls the behavior of the corrupt players arbitrarily throughout the protocol execution.

A stronger type of adversary is the \emph{adaptive} adversary, which may choose players to corrupt at any point during protocol execution,
and controls the behavior of the corrupt players arbitrarily from the moment of corruption until protocol termination.

\begin{definition}[Adaptive adversary]
Within each round, the adversary chooses players one-by-one to send their messages; 
and he can perform corruptions at any point during this process.
\end{definition}



\subsection{Security of coin-flipping protocols}

The security of a coin-flipping protocol is usually measured by the extent to which an adversary can, by corrupting a subset of parties,
bias the protocol outcome towards his desired bit.

\begin{definition}[$\eps$-security]
A coin-flipping protocol $\Pi$ is \emph{$\eps$-secure} against $t=t(n)$ adaptive (or static or strong adaptive) 
corruptions if for all $n\in\NN$, it holds that 
for any adaptive (resp. static or strong adaptive) adversary $\Adv$ that corrupts at most $t=t(n)$ players,
$$\min\left({\Pr}^{\Pi_n,\Adv}(0),{\Pr}^{\Pi_n,\Adv}(1)\right)\geq\eps.$$
\end{definition}

We remark that this definition of $\eps$-security is sometimes referred to as \emph{$\eps$-control} or \emph{$\eps$-resilience} in other works.
We next define a \emph{secure} protocol to be one
with ``minimal'' security properties (that is, one where the adversary does not almost always get the outcome he wants).

\begin{definition}[Security]
A coin-flipping protocol is \emph{secure} against $t=t(n)$ corruptions 
if it is $\eps$-secure against $t$ corruptions for some constant $0<\eps<1$. 
\end{definition}

In this work, we investigate the maximum proportion of adaptive corruptions that can be tolerated by \emph{any} secure protocol.

\section{Our results}

\subsection{Strongly adaptive adversaries}

In this work, we propose a new, stronger adversarial model than those that have been studied thus far (see section \ref{sec:advModels}), 
in which the adversary can see all honest players' messages within any given round, 
and \emph{subsequently} decide which players to corrupt. That is, he can see all the messages that the honest players ``would have sent'' in a round,
and then selectively intercept and alter these messages.

\begin{definition}[Strong adaptive adversary]
Within each round, the adversary sees all the messages that honest players would have sent, 
then gets to choose which (if any) of those messages to corrupt (i.e. replace with messages of his choice).
\end{definition}

This notion is an essential tool underlying the proof techniques in our work.
Moreover, we believe that the notion of \emph{strong adaptive security} gives rise to a natural and interesting new adversarial model
in which to study multi-party protocols, which is of independent interest beyond the scope of this work.

\subsection{Corruption tolerance in secure coin-flipping protocols}

Our main contributions consist of the following three results.
These can be viewed as partial progress towards proving the 30-year-old conjecture of \cite{BL85}.

\begin{theorem}\label{thm:main}
Any one-round coin-flipping protocol $\Pi$ can be secure against at most $t=\widetilde{O}(\sqrt n)$ strong adaptive corruptions.
\end{theorem}

\begin{theorem}\label{thm:minmax}
For any symmetric one-round coin-flipping protocol $\Pi$ secure against $t=t(n)$ adaptive corruptions,
there is a symmetric one-round coin-flipping protocol $\Pi'$ secure against $\Omega(t)$ strong adaptive corruptions.
\end{theorem}

\begin{corollary}
Any symmetric one-round coin-flipping protocol $\Pi$ can be secure against at most $t=\widetilde{O}(\sqrt n)$ adaptive corruptions.
\end{corollary}

In the next sections, we proceed to give detailed proofs of the theorems.

\subsection{Proof of Theorem \ref{thm:main}}

We begin by recalling the result of Lichtenstein et al. \cite{LLS89} which proves that the maximum number of adaptive corruptions for any secure \emph{single-bit} 
coin-flipping protocol is $O(\sqrt n)$. Note that the \emph{majority protocol} is the one-round protocol in which each player broadcasts a random bit,
and the majority of broadcasted bits is taken to be the protocol outcome.

\begin{theorem}[\cite{LLS89}]\label{thm:lls}
Any coin-flipping protocol in which each player broadcasts at most one bit can be secure against at most $t=O(\sqrt n)$ corruptions.
Moreover, the majority protocol achieves this bound.
\end{theorem}

Next, we establish some definitions and supporting lemmas.

\newcommand{\dist}{{\sf dist}}
\begin{definition}[Distance between message-vectors]
For vectors $\vec{r},\vec{r'}\in\cM^n$, let $\dist(\vec{r},\vec{r'})$ be equal to the number of coordinates $i\in[n]$
for which $r_i\neq r'_i$.
\end{definition}

\newcommand{\Robust}{{\sf Robust}}
\begin{definition}[Robust sets]
Let $\Pi$ be a one-round coin-flipping protocol in which each player sends a message from a message space $\cM$.
For any $n\in\NN$ and $b\in\{0,1\}$, define the set $\Robust^{\Pi_n}(b,t)$ as follows:
$$\Robust^{\Pi_n}(b,t)=\left\{\vec{r}\in\cM^n ~:~ \forall\vec{r'}\in\cM^n \mbox{ s.t. } \dist(\vec{r},\vec{r'})\leq t, ~\Pi_n(\vec{r})=\Pi_n(\vec{r'})=b\right\}.$$
\end{definition}

\begin{lemma}\label{lem:robustSets}
Let $\Pi$ be a one-round coin-flipping protocol in which each player sends a random message from a message space $\cM$.
$\Pi$ is secure against $t=t(n)$ strong adaptive corruptions if and only if
there exists a constant $0<\eps<1$ such that for all $n\in\NN$ and each $b\in\{0,1\}$,
$$\Pr_{\vec{r}\larr\cM}\left[\vec{r}\in\Robust^{\Pi_n}(b,t)\right]\geq\eps.$$
\end{lemma}
\begin{proof}
{\sc (``if'')}
Suppose that there exists a constant $0<\eps<1$ such that for all $n\in\NN$ and all $b\in\{0,1\}$, it holds that
\begin{align}\label{eqn:probGoodConstant}
\Pr_{\vec{r}\larr\cM^n}\left[\vec{r}\in\Robust^{\Pi_n}(b,t)\right]\geq\eps.
\end{align}

Let $\Adv$ be any strong adaptive adversary making up to $t$ corruptions.
For $n$-vector of (honest) messages $\vec{r}\in\cM^n$, let $\Adv(\vec{r})\in\cM^n$ denote the 
corresponding corrupted message-vector, where up to $t$ of the messages have been modified by $\Adv$.
By the definition of the set $\Robust^{\Pi_n}(b,t)$, 
it holds that
\begin{align}\label{eqn:advLoses}
\Pr_{\vec{r}\larr\cM^n}\left[\Pi_n(\Adv(\vec{r}))=b ~|~ \vec{r}\in\Robust^{\Pi_n}(b,t)\right]=1.
\end{align}

Combining equations \eqref{eqn:probGoodConstant} and \eqref{eqn:advLoses}, it follows that for each outcome $b\in\{0,1\}$,
$$
\Pr_{\vec{r}\larr\cM^n}\left[\Pi_n(\Adv\left(\vec{r})\right)=b\right]\geq\eps.
$$

We have shown that for each $b\in\{0,1\}$, $\Pr^{\Pi,\Adv}(b)\geq\eps$, as required. 

{\sc (``only if'')}
Suppose, on the other hand, that there is no constant $0<\eps<1$ such that for all $b\in\{0,1\}$, it holds that
$\Pr_{\vec{r}\larr\cM^n}\left[\vec{r}\in\Robust^{\Pi_n}(b,t)\right]=\eps$.
That is, there exists some $\eps'=o(1)$ such that for some $b\in\{0,1\}$ and infinitely many values of $n\in\NN$, it holds that
\begin{align}\label{eqn:fewGoods}
\Pr_{\vec{r}\larr\cM^n}\left[\vec{r}\in\Robust^{\Pi_n}(b,t)\right]\leq \eps'.
\end{align}
Without loss of generality, let $b=0$ be the bit for which equation \eqref{eqn:fewGoods} holds.
By the definition of $\Robust^{\Pi_n}(b,t)$, it holds that for any $\vec{r}\notin\Robust^{\Pi_n}(b,t)$,
there exists a vector $\vec{r}_{bad}\in\cM^n$ such that $\dist(\vec{r},\vec{r}_{bad})\leq t$ and $\Pi_n(\vec{r})\neq\Pi_n(\vec{r}_{bad})$.
In other words, if the honest players' messages $\vec{r}$ do not fall in $\Robust^{\Pi_n}(0,t)$, then
it is possible for a strong adaptive adversary $\Adv$ to \emph{force} the outcome to be 1, by doing as follows:
$$\Adv(\vec{r}) = \begin{cases}
\vec{r} & \mbox{ if } \Pi_n(\vec{r})=1 \\
\vec{r}_{bad} & \mbox{ if } \Pi_n(\vec{r})=0
\end{cases}$$
Note that since $\dist(\vec{r},\vec{r}_{bad})\leq t$, it is always possible for the adversary to change from $\vec{r}$ to $\vec{r}_{bad}$
using $t$ or fewer corruptions. Moreover, if $\Pi_n(\vec{r})=0$, then it must be that $\Pi_n(\vec{r}_{bad})=1$, by construction of $\vec{r}_{bad}$.
Hence, 
\begin{align}\label{eqn:advWins}
\Pr_{\vec{r}\larr\cM^n}\left[\Pi_n(\Adv(\vec{r}))=1 ~|~ \vec{r}\notin\Robust^{\Pi_n}(0,t)\right]=1.
\end{align}
Combining equations \eqref{eqn:fewGoods} and \eqref{eqn:advWins} (for $b=0$), we obtain:
$$
\Pr_{\vec{r}\larr\cM^n}\left[\Pi_n(\Adv(\vec{r}))=1\right] = 
\Pr_{\vec{r}\larr\cM^n}\left[\vec{r}\notin\Robust^{\Pi_n}(0,t)\right] \geq 1-\eps'.
$$
Hence, $\Pr^{\Pi,\Adv}(1)\geq 1-\eps'$, and so $\Pr^{\Pi,\Adv}(0)\leq\eps'=o(1)$.
Therefore, $\Pi$ is not secure against $t$ strong adaptive corruptions.
The lemma follows.
\end{proof}

Since players are computationally unbounded and we consider one-round protocols,
we may without loss of generality consider public-coin protocols\footnote{This is without loss of generality: 
each player can simply send his random coin tosses, and security holds since we are in the full-information model.}: 
for any one-round protocol $\Pi$ in the full-information model,
there is a protocol $\Pi'$ with an identical output distribution (in the presence of any adversary),
in which honest players send random messages in $\{0,1\}^k$ for some $k=\poly(n)$.

The following lemma serves as a stepping-stone to our final theorem. 

\begin{lemma}\label{lem:manyBitsMatrix}
For any one-round multi-bit coin-flipping protocol $\Pi$ secure against $t=t(n)$ strong adaptive corruptions, and any constant $\delta>0$,
there is a one-round $\ell$-bit coin-flipping protocol $\Pi'$ that is secure against $t$ strong adaptive corruptions,
where $\ell=O(\log^{1+\delta}(n))$.
\end{lemma}
\begin{proof}
Without loss of generality, we consider only public-coin protocols,
and assume that each player sends a message of the same length (say, $k=k(n)$ bits).
Let $\delta>0$ be any constant, let $\ell=O(\log^{1+\delta}(n))$, and let $\ell'=2^{\ell}$.

For an $\ell'\times n$ matrix of messages $M\in(\{0,1\}^k)^{\ell'\times n}$, we define the protocol $\Pi^M$ as follows:
each player $P_i$ broadcasts a random integer $a_i\larr[\ell']$, and the protocol outcome is defined by
$$\Pi^M_n(a_1,\dots,a_n) = \Pi_n(M_{(a_1,1)},\dots,M_{(a_n,n)}),$$
where $M_{(i,j)}$ denotes the message at the $i^{th}$ row and $j^{th}$ column of the matrix $M$.
For notational convenience, define $\vec{M}(a_1,\dots,a_n)=(M_{(a_1,1)},\dots,M_{(a_n,n)})$.
Notice that by construction of the protocol $\Pi^M$, it holds that for any message-vector $\vec{a}\in[\ell']^n$,
\begin{align}\label{eqn:robustSetsSame}
\vec{M}(\vec{a})\in\Robust^{\Pi_n}(b,t) ~\Longrightarrow~ \vec{a}\in\Robust^{\Pi^M_n}(b,t).
\end{align}

Suppose each entry of the matrix $M$ is a uniformly random message in $\{0,1\}^k$.
Note that the length of each player's message in $\Pi^M$ is $\log(\ell')=\ell$.
We want to show that $\Pi^M$ is a secure coin-flipping protocol against $t$ strong adaptive corruptions, for some $M$.
By Lemma~\ref{lem:robustSets}, it is sufficient to show that there exists $M\in(\{0,1\}^k)^{\ell'\times n}$ such that for all $b\in\{0,1\}$,
\begin{align}\label{eqn:notReallyRandom1}
\Pr_{\vec{a}\larr[\ell']^n}\left[\vec{a}\in\Robust^{\Pi^M_n}(b,t)\right]\geq\eps,
\end{align}
where $0<\eps<1$ is constant.
Using implication \eqref{eqn:robustSetsSame}, it actually suffices to prove:
\begin{align}\label{eqn:notReallyRandom2}
\exists M\in(\{0,1\}^k)^{\ell'\times n} \mbox{ s.t. } \forall b\in\{0,1\},~~
\Pr_{\vec{a}\larr[\ell']^n}\left[\vec{M}(\vec{a})\in\Robust^{\Pi_n}(b,t)\right]\geq\eps.
\end{align}

Suppose the matrix $M$ is chosen uniformly at random. Let $\vec{a}_1,\dots\vec{a}_n$ be sampled independently and uniformly from $[\ell']^n$.
Since, the number of matrix rows $\ell'=2^{O(\log^{1+\delta}(n))}$ is super-polynomial,
it is overwhelmingly likely that $\vec{a}_1,\dots\vec{a}_n$ will be composed of distinct elements in $[\ell']$. 
That is, to be precise, 
$$\Pr_{\vec{a}_1,\dots,\vec{a}_n}\left[\forall (i,j)\neq(i',j')\in[n]\times[n], ~ (\vec{a}_i)_j \neq (\vec{a}_{i'})_{j'}\right]\geq 1-\negl(n).$$
If $\vec{a}_1,\dots,\vec{a}_n$ are indeed composed of distinct elements, the message-vectors
$\vec{M}(\vec{a}_1),\dots,\vec{M}(\vec{a}_n)$ are independent random elements in $(\{0,1\}^k)^n$.
Thus,
\begin{align}\label{eqn:statIndist}
(\vec{M}(\vec{a}_1),\dots,\vec{M}(\vec{a}_n))\statIndist(\vec{r}_1,\dots,\vec{r}_n),
\end{align}
when $M$ is a random matrix in $(\{0,1\}^k)^{\ell'\times n}$, 
the (short) message-vectors $\vec{a}_1,\dots,\vec{a}_n$ are random in $[\ell']^n$, 
and the (long) message-vectors $\vec{r}_1,\dots,\vec{r}_n$ are random in $(\{0,1\}^k)^n$.

Since $\Pi$ is a secure coin-flipping protocol, there is a constant $0<\eps'<1$ such that for all $n\in\NN$ and $b\in\{0,1\}$ and $i\in[n]$,
$$\Pr_{\vec{r}_i}\left[\vec{r}_i\in\Robust^{\Pi_n}(b,t)\right]\geq\eps'.$$
The rest of the proof follows from a series of Chernoff bounds.

For $i\in[n]$ and $b\in\{0,1\}$, let $Z_{i,b}$ be an indicator variable for the event that $\vec{r}_i\in\Robust^{\Pi_n}(b,t)$.
Since the $\vec{r}_i$ are independent, we apply a Chernoff bound to obtain the following (for all $b\in\{0,1\}$):
\begin{align}\label{eqn:chernoff}
\Pr_{\vec{r}_1,\dots,\vec{r}_n}\left[\frac{1}{n}\cdot\sum_{i\in[n]}Z_{i,b}<\eps'- \eps''\right]\leq \negl(n),
\end{align}
for any constant $0<\eps''<\eps'$. 

Let $Y_{i,b}$ be an indicator variable for the event that $\vec{M}(\vec{a}_i)\in\Robust^{\Pi_n}(b,t)$.
It follows from \eqref{eqn:statIndist} and \eqref{eqn:chernoff} that with overwhelming probability over the choice of the random matrix $M$, 
it holds for all $b\in\{0,1\}$ that 
\begin{align}\label{eqn:chernoffSet}
\Pr_{\vec{a}_1,\dots,\vec{a}_n}\left[\frac{1}{n}\cdot\sum_{i\in[n]}Y_{i,b}<\eps'- \eps''\right]\leq \negl(n).
\end{align}

For $b\in\{0,1\}$, let $\alpha_b$ denote the probability $\Pr_{\vec{a}_i}\left[\vec{M}(\vec{a}_i)\in\Robust^{\Pi_n}(b,t)\right]$.
Note that for any given $b\in\{0,1\}$ the variables $Y_{i,b}$ are independently and identically distributed, each taking value 1 with probability $\alpha_b$
and value 0 with probability $1-\alpha_b$. 
By a Chernoff bound, for any constant $0<\eps'''<1$, it holds that (with overwhelming probability over the choice of $M$):
\begin{align}\label{eqn:realChernoff}
\Pr_{\vec{a}_1,\dots,\vec{a}_n}\left[\left|\frac{1}{n}\cdot\sum_{i\in[n]}Y_{i,b}-\alpha_b\right|\geq\eps'''\right]\leq\negl(n).
\end{align}
From \eqref{eqn:chernoffSet} and \eqref{eqn:realChernoff}, it follows that with overwhelming probability over the random choice of $M$, 
for all $b\in\{0,1\}$ and any constant $0<\eps''<1$ and $0<\eps'''<1$,
$$
\Pr_{\vec{a}_1,\dots,\vec{a}_n}\left[\alpha_b<\eps'-\eps''-\eps'''\right]\leq\negl(n).
$$
By taking $\eps''+\eps'''\leq\eps'/2$, we have that with overwhelming probability over $M$, it holds that $\alpha_b<\eps'/2$ for all $b\in\{0,1\}$.
Finally, the $\alpha_b$ correspond exactly to the probability expression in \eqref{eqn:notReallyRandom2}, so
we have shown statement \eqref{eqn:notReallyRandom2} as required.
\end{proof}

Having reduced the length of players' messages to $\polylog(n)$ in Lemma \ref{lem:manyBitsMatrix},
we now prove the following lemma which reduces the required communication even further, so that each player sends only one bit.
This comes at the cost of a polylogarithmic factor reduction in the number of corruptions.

Before the lemma, we recall the statement of the Chernoff bound.

\begin{theorem}[Chernoff bound]
Let $X_1,\dots,X_n$ be independent random variables taking values in $\{0,1\}$, which all have the same expectation $\mu=\EE[X_i]$. Then, for every $0<\eps<1$,
$$\Pr\left[\left|\frac{1}{n}\cdot\sum_{i\in[n]} X_i-\mu\right|\geq\eps\right]\leq 2e^{-2n\eps^2}.$$
\end{theorem}

\begin{lemma}\label{lem:multiToOneBit_strong}
For any one-round $\ell$-bit coin-flipping protocol $\Pi$ secure against $t=t(n)$ strong adaptive corruptions, 
there is a one-round single-bit coin-flipping protocol $\Pi'$ that is secure against $t/\ell$ strong adaptive corruptions.
\end{lemma}
\begin{proof}
Let $\Pi$ be any one-round $\ell$-bit coin-flipping protocol secure against $t=t(n)$ strong adaptive corruptions.
We define our new single-bit protocol\footnote{We remark that the protocol $\Pi'$ that we construct does not strictly adhere to
Definition~\ref{def:coinFlipping}, because $\Pi'=\{\Pi_n\}_{n\in\ell\cdot\NN}$ does not define an $n$-player protocol for every $n\in\NN$.
We consider this to be a very minor technical detail that we bury for clarity of exposition.
} $\Pi'$ as follows, for each $n\in\NN$:
\begin{gather*}
\Pi'_{n\cdot\ell}(r_1,\dots,r_{n\cdot\ell})= \\
\Pi_n\left((r_1||\dots||r_\ell), (r_{\ell+1}||\dots||r_{2\ell}), \dots, (r_{(n-1)\cdot\ell+1}||\dots||r_{n\cdot\ell})\right),
\end{gather*}
where the messages $r_i\in\{0,1\}$ are bits and $||$ denotes concatenation.
Informally speaking, there are $n$ groups of $\ell$ players in the single-bit protocol $\Pi'_{n\cdot\ell}$, each of which
``corresponds to'' a single player in the protocol $\Pi_n$.

We show that $\Pi'$ is secure against $t/\ell$ corruptions.
Let $G_i$ denote the $i^{th}$ group of $\ell$ players: to be precise, $G_i=\{i\cdot\ell+1,\dots,(i+1)\cdot\ell\}$.
If all of the players in the set $G_i$ are honest, then the $i^{th}$ ``combined message'' $(r_{i\cdot\ell+1}||\dots||r_{(i+1)\cdot\ell})$
is distributed identically to an honest message of the $i^{th}$ player in the protocol $\Pi_n$.
By the construction of the protocol $\Pi'$, it follows that for any $b\in\{0,1\}$ and $n\in\NN$,
\begin{align}\label{eqn:robustSingleBit}
\Pr_{\vec{r}\larr\{0,1\}^{n\cdot\ell}}\left[\vec{r}\in\Robust^{\Pi'_{n\cdot\ell}}(b,t(n))\right]\geq
\Pr_{\vec{r'}\larr(\{0,1\}^\ell)^n}\left[\vec{r'}\in\Robust^{\Pi_n}(b,t(n))\right].
\end{align}
By Lemma \ref{lem:robustSets}, since $\Pi$ is secure against $t$ strong adaptive corruptions,
there is a constant $0<\eps<1$ such that for all $b\in\{0,1\}$ and $n\in\NN$, the right-hand side of inequality \eqref{eqn:robustSingleBit}
is at least $\eps$. Hence we obtain
$$
\Pr_{\vec{r}\larr\{0,1\}^{n\cdot\ell}}\left[\vec{r}\in\Robust^{\Pi'_{n\cdot\ell}}(b,t(n))\right]\geq\eps.
$$
It follows (by applying Lemma \ref{lem:robustSets} again) that $\Pi'$ is secure against $t/\ell$ strong adaptive corruptions.
\end{proof}

Finally, we bring together Lemmas \ref{lem:manyBitsMatrix} and \ref{lem:multiToOneBit_strong} to prove the theorem.

\begin{customthm}{\ref{thm:main}}
Any one-round coin-flipping protocol $\Pi$ can be secure against at most $t=\widetilde{O}(\sqrt n)$ strong adaptive corruptions.
\end{customthm}
\begin{proof}
Suppose, for contradiction, that there exists a one-round coin-flipping protocol $\Pi$ which is secure against $t$ corruptions,
where $t=\omega(\sqrt{n}\cdot\polylog(n))$.
Then, by Lemma \ref{lem:manyBitsMatrix}, there is an $\ell$-bit one-round coin-flipping protocol $\Pi'$
that is secure against $t$ strong adaptive corruptions, where $\ell=\polylog(n)$. 
By applying Lemma \ref{lem:multiToOneBit_strong} to the protocol $\Pi'$,
we deduce that there is a single-bit one-round coin-flipping protocol $\Pi''$ which is secure against $t/\ell=\widetilde{\Omega}(t)$ strong adaptive corruptions.
Since a \emph{strongly adaptive} adversary can perfectly simulate
any strategy of an \emph{adaptive} adversary, it follows that $\Pi''$ is secure against $\widetilde{\Omega}(t)$ adaptive corruptions.
Since $\Pi''$ is single-bit, this contradicts Theorem \ref{thm:lls}.
\end{proof}

\subsection{Proof of Theorem \ref{thm:minmax}}

In this section, we show that for any symmetric one-round coin-flipping protocol secure against $t$ \emph{adaptive} corruptions,
there is a one-round coin-flipping protocol secure against $\Omega(t)$ corruptions by \emph{strong adaptive} adversaries.
That is, one-round strong adaptively secure protocols are a more general class than one-round symmetric, adaptively secure protocols.

\begin{remark}
In fact, Theorem \ref{thm:minmax} holds even if the protocol $\Pi$ is just statically secure: the proof does not make use of the fact
that $\Pi$ is adaptively, rather than statically, secure. 
Our theorem statement refers to $\Pi$ as an adaptively secure protocol because this is exactly what we need in order to obtain our final result
that any one-round symmetric coin-flipping protocol can be secure against at most $O(\sqrt{n})$ corruptions.
\end{remark}

The Minimax Theorem -- a classic tool in game theory -- will be an important tool in our proof. 
The statement of the Minimax Theorem and supporting game-theoretic definitions are given below.

\begin{definition}[Two-player strategic game]
A \emph{two-player finite strategic game} $\Gamma=\langle (A_1,A_2),(u_1,u_2)\rangle$ is defined by:
for each player $i\in\{1,2\}$,
a non-empty set of possible \emph{actions} $A_i$ and
a \emph{utility function} $u_i:A_1\times A_2\rarr\RR$.
\end{definition}

\begin{definition}[Zero-sum game]
A two-player finite strategic game $\Gamma=\langle (A_1,A_2),(u_1,u_2)\rangle$ is \emph{zero-sum} if
for any pair of actions $a_1\in A_1$ and $a_2\in A_2$, it holds that 
$u_1(a_1,a_2)+u_2(a_1,a_2)=0$.
\end{definition}

\begin{theorem}[Minimax \cite{vN44,Nash}]\label{thm:originalMinMax}
Let $\Gamma=\langle (A_1,A_2),(u_1,u_2)\rangle$ be a zero-sum two-player finite strategic game.
Then 
$$\max_{a_2\in \Delta(A_2)} \min_{a_1\in \Delta(A_1)} u_2(a_1,a_2) = \min_{a_1\in \Delta(A_1)} \max_{a_2\in \Delta(A_2)} u_1(a_1,a_2),$$
where $\Delta(A_i)$ denotes the set of distributions over $A_i$ (in game-theoretic terminology, this corresponds to the set of ``mixed strategies'' for player $i$.)
\end{theorem}

\begin{customthm}{\ref{thm:minmax}}
For any symmetric one-round coin-flipping protocol $\Pi$ secure against $t=t(n)$ adaptive corruptions,
there is a symmetric one-round coin-flipping protocol $\Pi'$ secure against $s=t/2$ strong adaptive corruptions.
\end{customthm}
\begin{proof}
Let $\Pi$ be a symmetric one-round coin-flipping protocol secure against $t=t(n)$ adaptive corruptions,
and define $s(n)=t(n)/2$.
We define a new protocol $\Pi'=\{\Pi'_n\}_{n\in\NN}$ as follows:
$$\Pi'_n(r_1,\dots,r_n)=\min_{\rprimesblue}\max_{\rprimesred}\Pi_{n+2s}\left(\rs{1}{n},\rprimesblue,\rprimesred\right),$$
where $s=s(n)$ and honest players in $\Pi'_n$ must send messages according to the same distributions as in $\Pi_{n+2s}$.

Observe that $\Pi_{n+2s}$ is secure against $t(n+2s(n))>t(n)$ corruptions. 
We show that $\Pi'_n$ is secure against $s(n)=t(n)/2$ strong adaptive corruptions.

{\sc Case 1.} Suppose that the adversary aims to bias the outcome towards $0$. 
By the security of $\Pi_{n+2s}$, there is a constant $0<\eps<1$ such that
$\Pr^{\Pi_{n+2s},\Adv}(1)\geq\eps$ for any adaptive adversary $\Adv$ that corrupts up to $t=2s$ players. 
Without loss of generality (since the protocol is symmetric), suppose that the adversary corrupts the last $2s$ players in $\Pi_{n+2s}$.

We say that the honest players' messages $\rs{1}{n}$ ``fix'' the outcome of
$\Pi_{n+2s}$ to be $1$ if for any possibly malicious messages $\rhats{1}{2s}$,
it holds that $\Pi_{n+2s}(\rs{1}{n},\rhats{1}{2s})=1$.
Then, with probability at least $\eps$, the honest players' messages $\rs{1}{n}$ ``fix'' the outcome of $\Pi_{n+2s}$ to be $1$.
(To see this: suppose not. Then there would exist an adversary which could set the corrupt messages $\rhats{1}{2s}$ 
so that the protocol outcome is $0$ with probability $1-\eps$.
But this cannot be, since we already established that $\Pr^{\Pi_{n+2s},\Adv}(1)\geq\eps$.)

Define the set
$R_1\defeq\left\{(\rs{1}{n}):\forall\rhats{1}{2s},~~\Pi_{n+2s}(\rs{1}{n},\rhats{1}{2s})=1\right\}$
to consist of those honest message-vectors that fix the output of $\Pi_{n+2s}$ to be $1$.

Take any $(\rs{1}{n})\in R_1$. 
We now show that the outcome of $\Pi'_n$ when the honest players send messages $\rs{1}{n}$ is equal to $1$,
even in the presence of a strong adaptive adversary $\Adv'$ that corrupts up to $s$ players and aims to bias the outcome towards $0$.
Without loss of generality, suppose that $\Adv'$ corrupts the first $s$ players in $\Pi'_n$, and replaces their honest messages $\rs{1}{s}$
with some maliciously chosen messages $\rhats{1}{s}$. 
In this case, the outcome of $\Pi'_n$ is
\begin{align*}
\Pi'_n&(\hat{r}_1,\dots,\hat{r}_s,r_{s+1},\dots,r_n) \\
& =\min_{\rprimesblue}\max_{\rprimesred}\Pi_{n+2s}\left(\rhats{1}{s},\rs{s+1}{n},\rprimesblue,\rprimesred\right) \\
& \geq\min_{\rprimesblue}\Pi_{n+2s}\left(\rhats{1}{s},\rs{s+1}{n},\rprimesblue,\rs{1}{s}\right) \\
& =\min_{\rprimesblue}\Pi_{n+2s}\left(\rs{1}{n},\rhats{1}{s},\rprimesblue\right)\tag{by symmetry} \\
& =1,
\end{align*}
where the last line follows from the definition of $R_1$, since we started with $(\rs{1}{n})\in R_1$.

We already established that the probability that the honest players' messages fall in $R_1$ is at least $\eps$.
Thus we deduce that with probability at least $\eps$, the outcome of the new protocol $\Pi'_n$ is equal to $1$, 
even in the presence of a strong adaptive adversary corrupting $s$ players and aiming to bias towards $0$.

{\sc Case 2.} Suppose instead that the adversary $\Adv'$ aims to bias the outcome towards $1$.
We apply the Minimax Theorem to a zero-sum game where player 1 chooses the messages $\rprimesblue$ and player 2 chooses the messages $\rprimesred$,
and player 1 ``wins'' if the protocol outcome is 0, and player 2 wins otherwise. 
By the Minimax Theorem,
$$\Pi'_n(r_1,\dots,r_n)=\max_{\rprimesred}\min_{\rprimesblue}\Pi_{n+2s}\left(\rs{1}{n},\rprimesblue,\rprimesred\right).$$

Given this new and equivalent definition of $\Pi'_n$, we can apply exactly 
the same argument structure as that given for Case 1 above, to deduce that
\begin{itemize}
\item There is a constant $0<\eps'<1$ such that $\Pr^{\Pi_{n+2s},\Adv}(0)=1-\Pr^{\Pi_{n+2s},\Adv}(1)=\eps'$ for any adaptive $\Adv$ performing up to $2s$ corruptions, 
and hence there is a non-empty set $$R_0\defeq\left\{(\rs{1}{n}):\forall\rhats{1}{2s},~~\Pi_{n+2s}(\rs{1}{n},\rhats{1}{2s})=0\right\}\mbox{,~~and}$$
\item by the adaptive security of $\Pi_{n+2s}$, the messages of honest players will fall in $R_0$ with probability at least $\eps'$, and
\item if the honest players' messages fall in $R_0$, then the outcome of $\Pi'_n$ is equal to $0$, 
even in the presence of a strong adaptive adversary corrupting $s$ players and aiming to bias towards $1$.
\end{itemize}

We have established that both outcomes 0 and 1 occur with constant probability in $\Pi'_n$, even in the presence of 
an arbitrary strong adaptive adversary corrupting up to $s$ players. 
Therefore, $\Pi'_n$ is secure against $s=t/2$ corruptions.
\end{proof}

\section{Conclusion}

We have introduced a new adversarial model for multi-party protocols and an associated security notion, \emph{strong adaptive security}.
We have made use of a novel and widely applicable technique for reducing the amount of communication in a protocol,
to show that any one-round strongly adaptively secure coin-flipping protocol can tolerate at most $\widetilde{O}(\sqrt n)$ corruptions.
We believe that this work paves the way to a number of little-explored research directions.
We highlight some interesting questions for future work:

\begin{itemize}
\item To study the extent to which \emph{communication can be reduced in protocols in general}, and to extend our communication-reduction techniques to
the settings of multi-round protocols and/or adaptive security.
\item To apply the \emph{strong adaptive security notion} in the context of other types of protocols and settings, 
and to design protocols secure in the presence of strong adaptive adversaries.
\item To consider whether adaptively secure \emph{asymmetric} coin-flipping protocols can be converted to adaptively secure \emph{symmetric} protocols, in general. This is not known even for the one-round case, and the question is moreover of interest since there are known one-round protocols which are not symmetric.
\item To extend this work to prove (or disprove) the long-open conjecture of Lichtenstein et al. \cite{LLS89} that \emph{any} adaptively secure coin-flipping protocol
can tolerate at most $O(\sqrt n)$ corruptions. 
\end{itemize}

\printbibliography

\end{document}